\documentclass{article}
\usepackage{ijcai23}

\usepackage{times}
\usepackage{soul}
\usepackage{url}
\usepackage[hidelinks]{hyperref}
\usepackage[utf8]{inputenc}
\usepackage[small]{caption}
\usepackage{graphicx}
\usepackage{amsmath}
\usepackage{amsthm}
\usepackage{booktabs}
\usepackage{algorithmic}
\usepackage[switch]{lineno}

\usepackage[dvipsnames]{xcolor}
\usepackage{color,epsfig}
\usepackage{pifont}
\usepackage{latexsym}
\usepackage{amsmath,amsfonts,amssymb,nicefrac}
\usepackage{booktabs,tabularx,multirow,caption}
\usepackage{graphicx,tikz}
\usepackage{todonotes}
\usepackage{skull}
\usepackage{multicol}
\usepackage{verbatim}
\usepackage[ruled,linesnumbered, noline]{algorithm2e}

\def \opt{{\mathsf{opt}}}

\def\E{{\mathbb E}}

\def\SMC{\mathsf{SMC}}
\def\SMK{\mathsf{SMK}}

\def\LA{\mathsf{LA}}

\def\LAR{\mathsf{LAR}}
\def\DLA{\mathsf{DLA}}
\def\RLA{\mathsf{RLA}}
\def\NSMK{\mathsf{\mbox{non-monotone } SMK}}
\def\Alg{\mathsf{Alg}}

\newtheorem{lemma}{Lemma}
\newtheorem{theorem}{Theorem}

\title{Linear Query Approximation Algorithms for Non-monotone Submodular Maximization under Knapsack Constraint}

\author{
	Canh V. Pham$^1$\footnote{Corresponding author}
	\and
	Tan D. Tran$^2$\and
	Dung T. K. Ha$^{2}$\And
	My T. Thai$^3$
	\affiliations
	$^1$ORLab, Faculty of Computer Science, Phenikaa University, Hanoi, Vietnam
	\\
	$^2$ Faculty of Information Technology, VNU University of Engineering and Technology, Hanoi, Vietnam\\
	$^3$Department of Computer and Information Science and Engineering
	\\
	University of Florida, Gainesville, Florida 32611
	\emails
	canh.phamvan@phenikaa-uni.edu.vn,
	\{22027005, 20028008\}@vnu.edu.vn, mythai@cise.ufl.edu
}
\begin{document}
	
	\maketitle
	
	\begin{abstract}
		This work, for the first time, introduces two constant factor approximation algorithms with linear query complexity for non-monotone submodular maximization over a ground set of size $n$ subject to a knapsack constraint, $\mathsf{DLA}$  and $\mathsf{RLA}$. $\mathsf{DLA}$ is a deterministic algorithm that provides an approximation factor of $6+\epsilon$ while $\mathsf{RLA}$ is a randomized algorithm with an approximation factor of $4+\epsilon$ in expectation. Both run in $O(n \log(1/\epsilon)/\epsilon)$ query complexity. The key idea to obtain a constant approximation ratio with linear query lies in: (1) dividing the ground set into two appropriate subsets to find the near-optimal solution over these subsets with linear queries, and (2) combining a threshold greedy with properties of two disjoint sets or a random selection process to improve solution quality. In addition to the theoretical analysis, we evaluated the original implementations with three applications: Revenue Maximization, Image Summarization, and Maximum Weighted Cut, showing that our algorithms not only return comparative results to state-of-the-art algorithms but also require significantly fewer queries. 
	\end{abstract}
	
	\section{Introduction}
	In the variety of submodular optimization, Submodular Maximization under a Knapsack ($\SMK$) constraint is one of the most fundamental problems. In this problem, given a ground set $V$ of size $n$ and a non-negative submodular set function $f: 2^V \mapsto \mathbb{R}_+$. Assume that each element $e \in V$  has a
	positive cost $c(e)$ and there is a budget $B$,  $\SMK$ asks for finding $S\subseteq V$ subject to $c(S)=\sum_{e \in S}c(e)\leq B$ that maximizes $f(S)$. $\SMK$ captures important constraints in practical applications, such as bounds on costs, time, or size, thereby attracting a lot of attention recently \cite{fast_icml,Amanatidis2021a_knap,Han2021_knap,sub_knap_orlet,li-linear-knap-nip22,Ene_HuyLNg_knap,Lee_nonmono_matroid_knap,Amanatidis_sampleGreedy,Gupta_nonmono_constrained_submax}. 
	
	In addition to obtaining a near-optimal solution to $\SMK$, designing such a solution also focuses on reducing query complexity, especially in an era of big data. With an explosion of input data, the search space for a solution has increased exponentially. Unfortunately, submodularity requires an algorithm to evaluate the objective function whenever observing an incoming element. Therefore, it is necessary to design efficient algorithms that reduce the number of queries to linear or nearly linear. 
	
	Furthermore, to model $\SMK$ for real-world applications, the objective functions may be non-monotone,
	since the marginal contribution of an element to a set may not always increase. Notable examples with non-monotone objective functions can be found in revenue maximization on  social network \cite{fast_icml,Kuhnle19-nip}, image summarization with a representative \cite{fast_icml} or maximum weight cut \cite{Amanatidis_sampleGreedy}.
	
	Unfortunately, no constant approximation algorithm with linear query complexity exists for non-monotone $\SMK$ compared to its counterpart. For the monotone $\SMK$, the best approximation factor of $e/(e-1)$ is achieved within $O(n^5)$ number of queries \cite{sub_knap_orlet}; and the fastest algorithm has a factor of $2+\epsilon$ needs a linear number of queries \cite{li-linear-knap-nip22}.  But for non-monotone, the best approximation algorithm needs polynomial queries and has a factor of $1/0.385$ \cite{BuchbinderF19-bestfactor} and the fastest algorithm with constant factor requires near-linear query complexity of $O(n\log k)$, where $k$ is the maximum cardinality of any feasible solution to $\SMK$ \cite{Han2021_knap}. Thus this work aims to close this gap by addressing the following open question: \textit{Is there a constant factor approximation algorithm for  non-monotone $\SMK$ in linear query complexity?}
	
	Solving non-monotone $\SMK$ with linear query complexity is more challenging than that of the monotone case due to the following reasons. First, the property of the monotone submodular function plays an important role in analyzing the theoretical bound of an obtained solution. Second, algorithms for the non-monotone case need more queries to obtain information from all elements in the condition that the marginal contribution of an element may be negative. 
	\begin{table*}[h]
		\centering
		\footnotesize{
			\begin{tabular}{cccccc}
				\hline
				\textbf{Reference} 
				& \textbf{Approximation factor} & \textbf{Query Complexity}& \textbf{Deterministic/Randomized}
				\\
				\hline
				\cite{fast_icml} (FANTOM)&$10+\epsilon$&$O(n^2\log(n)/\epsilon)$& Randomized
				\\
				\cite{Amanatidis_sampleGreedy} (SAMPLE GREEDY)&$5.83+\epsilon$&$O(n\log(n/\epsilon)/\epsilon)$& Randomized
				\\
				\cite{Han2021_knap} (SMKDETACC) &$6+\epsilon$&$O(n\log(k/\epsilon)/\epsilon)$& Deterministic
				\\
				\cite{Han2021_knap} (SMKSTREAM)  &$6+\epsilon$&$O(n\log (B)/\epsilon)$& Deterministic
				\\
				\cite{Han2021_knap} (SMKRANACC) &$4+\epsilon$&$O(n\log(k/\epsilon)/\epsilon)$& Randomized
				\\
				\hline
				\textbf{$\DLA$ (Algorithm~\ref{alg:dla}, this paper)}& $6+\epsilon$  & \textbf{$O(n \log(1/\epsilon)/\epsilon)$} & Deterministic
				\\
				\textbf{$\RLA$ (Algorithm~\ref{alg:rla}, this paper)}& $4+\epsilon$  & $O(n \log(1/\epsilon)/\epsilon)$ & Randomized
				\\
				\hline
			\end{tabular}
		}
		\caption{Fastest algorithms for $\NSMK$ problem, where $k$ is the maximum cardinality of any feasible solution to $\SMK$.}
		\label{tab:1}
	\end{table*}
	\paragraph{Our Contributions.} To tackle the above challenges, 		we propose two approximation algorithms, $\DLA$ and $\RLA$, that achieve a constant factor approximation, yet both require linear query complexity. 
	Our $\DLA$ is a deterministic algorithm with an approximation factor of $6+\epsilon$ within $O(n\log(1/\epsilon)/\epsilon)$ queries. Therefore, $\DLA$ is significantly faster than the deterministic algorithm of \cite{Han2021_knap}, which achieved the best-known approximation factor for $6+\epsilon$ with nearly-linear query complexity of $O(n\log(k/\epsilon)/\epsilon)$.
	$\RLA$ is a randomized algorithm that achieves an expected factor of $4+\epsilon$ in $O(n\log(1/\epsilon)/\epsilon)$ queries. Therefore, $\RLA$ achieves the same factor of the randomized algorithm as in \cite{Han2021_knap}, which currently provides the best approximation factor in near-linear query complexity of $O(n\log(k/\epsilon)/\epsilon)$.  Note that $k$ may be as large as $n$, so the query complexity of the algorithms in \cite{Han2021_knap} can be $O(n\log(n/\epsilon)/\epsilon)$. Table~\ref{tab:1} compares the performance of our algorithms with that of existing fast algorithms.
	
	Both our algorithms focus on a novel algorithmic approach that consists of two components: (1) dividing the ground set into two appropriate subsets and finding the approximation solution over these subsets with linear queries, and (2) combing the threshold greedy procedure
	developed by \cite{fast_sub} with two disjoint candidate solutions (for $\DLA$) or a random process (for $\RLA$) to construct serial candidate solutions to give better theoretical bounds. 
	At the heart of the first component, we adapt the method of simultaneously constructing two disjoint sets, which is first introduced by \cite{Han_twinGreedy,SimuGre_AmanatidisKS22} and later used by \cite{best-dla,Han2021_knap} to bound the utility of candidate solutions. By incorporating a method of dividing the ground set into two reasonable subsets, we can bound the cost of feasible solutions, thereby obtaining a constant approximation factor within only a one-time scan over these subsets. In the second component, we adapt the threshold greedy, where thresholds are adjusted accordingly to provide a constant number of candidate solutions. Finally, we boost the solution quality of our algorithm by re-scanning the best elements for selecting candidate solutions without increasing query complexity. 
	
	Experiments with the original implementations show that they outperform several state-of-the-art algorithms \cite{fast_icml,Amanatidis2021a_knap,Han2021_knap} regarding solution quality and the number of queries. In particular, $\DLA$ provides the best solution quality and needs fewer queries than the faster approximation deterministic algorithm in \cite{Han2021_knap}, while $\RLA$ returns competitive solutions but needs the fewest queries.
	\paragraph{Paper Organization.} The rest of the paper is structured as follows. 
	Section~\ref{sec:relatedwork} provides the literature review on $\NSMK$ problem. Notations are presented in Section~\ref{sec:preli}. Section~\ref{sec:algs} introduces our proposed algorithms and theoretical analysis.  Experimental computation is provided in Section~\ref{sec:expr}. Finally, we conclude this work in Section~\ref{sec:con}.
	\section{Related Works}
	\label{sec:relatedwork}
In this section, we review the related work for the $\NSMK$ problem only. A brief review of monotone $\SMK$ and submodular maximization subject to cardinality, a  special cases of $\SMK$, can be found in the Appendix. 

	Randomization is one of the effective methods for designing approximation algorithms for submodular $\NSMK$. The first randomized algorithm was proposed by \cite{lee-car-10} with a factor of $5+ \epsilon$; the factor was later improved to $4+\epsilon$ by \cite{Kulik13_Submax_knap}. Several researchers tried to enhance the approximation factor to $e/(e-1)+\epsilon$ or $e+\epsilon$  \cite{ChekuriVZ14,ran-smk-Felman11,Ene_HuyLNg_knap,BuchbinderF19-bestfactor}. The best factor in this line of randomized algorithms was $1/0.385 \approx 2.6$ due to \cite{BuchbinderF19-bestfactor}, using the multi-linear extension method with the rounding scheme technique in \cite{Kulik13_Submax_knap}. However, this work has to handle complicated multi-linear extensions and uses a large number of queries. In contrast, 
	\cite{Amanatidis_sampleGreedy} proposed a sample greedy, a fast algorithm with a factor of $5.83+\epsilon$ requiring $O(n \log(n/\epsilon)/\epsilon)$ queries. An efficient parallel algorithm  with a factor of $9.465+\epsilon$ was introduced by \cite{Amanatidis2021a_knap}, but it needed a high query complexity of $O(n^2\log^2(n)\log(1/\epsilon)/\epsilon^3)$.
	Significantly, \cite{Han2021_knap} introduced the current fastest randomized algorithm with the factor of $4+\epsilon$ in $O(n \log(k/\epsilon)/\epsilon)$ queries.
	
	For the deterministic algorithm approach,  \cite{Gupta_nonmono_constrained_submax} first presented a deterministic algorithm with a factor of $6$. Their algorithm modified Sviridenko’s algorithm \cite{sub_knap_orlet} and combined with an algorithm for unconstrained non-monotone submodular maximization \cite{usm}; however, it took $O(n^5)$ query complexity. Since then, there are several algorithms have been proposed to reduce the number of queries. The FANTOM algorithm \cite{fast_icml} improved the query complexity to $O(n^2 \log(n)/\epsilon)$ but returned a larger factor of $10$. Algorithm of \cite{nearly-liner-der-2018} achieved a factor of $9.5+ \epsilon$ in $O(nk)\max\{\epsilon^{-1}, \log \log n\}$ and it can be used for  $p$-system and $d$-knapsack constraints.   \cite{Cui-streaming21} introduced a streaming algorithm with a factor of $2.05+ \rho_{\Alg}$ in $O((n+T_{\Alg(k)})\log B)$ queries, where $\rho_{\Alg}$ was the approximation factor any offline algorithm $\Alg$ for $\SMK$  and $T_{\Alg(k)}$ was the query complexity of $\Alg$ with $k$ input elements. The factor and query complexity of the algorithm are quite large because they depend on $\rho_{\Alg}$ and $k$ can be as large as $n$.
	Recently, \cite{Han2021_knap} also presented another one that was deterministic the factor of $6+\epsilon$ in nearly-linear queries $O(n\log (k/\epsilon)/\epsilon)$. Currently, the best approximation factor of a deterministic algorithm for non-monotone $\SMK$ is due to \cite{best-dla} achieving an approximation factor of $4+\epsilon$ but requiring an impractical query complexity of $O(n^3 \log(n/\epsilon)/\epsilon)$.
	\section{Preliminaries}
	\label{sec:preli}
	We use the definition of submodularity based on \textit{the diminishing return property}: A set function $f: 2^V \mapsto \mathbb{R}_+ $, defined on all subsets of a ground set $V$ of size $n$  is  submodular iff  for any $A \subseteq B \subseteq V$ and $e \in V\setminus B$, we have: 
	\begin{align*}
	f(A \cup \{e\}) - f(A) \geq f(B \cup \{e\})- f(B).
	\end{align*}
	Each element $e \in V$ is assigned a positive cost $c(e)>0$, and the total cost of a set $S\subseteq V$ is a modular function, i.e., $c(S) =\sum_{e \in S}c(e)$. Given a budget $B$, we assume that every item $e\in V$ satisfies $c(e) \leq  B$; otherwise, we can simply discard it. The $\SMK$ problem is to determine: 
	\begin{align}
	\arg\max_{S\subseteq V:c(S)\leq B}f(S).
	\end{align}
	We denote an instance of $\SMK$ by a tuple $(f, V, B)$. For simplicity, we assume that $f$ is non-negative, i.e., $f(X)\geq 0$ for all $X\subseteq V$ and  normalized, i.e., $f(\emptyset)=0$. We define the contribution gain of an element $e$ to a set $S\subseteq V$ as $f(e|S)=f(S\cup \{e\})-f(S)$ and we  write $f(\{e\})$ as $f(e)$ for
	any $e\in V$.
	We assume that there exists an oracle query, which when queried with the
	set $S$ returns the value $f(S)$.
	All scans use a fixed order and deterministic tie-breaking unless a random choice is specified. Cache the singleton values computed on the first scan. If no feasible item remains or the best singleton has value zero, return $\emptyset$ before starting any threshold loop; the latter case has $\opt=0$ since $f(S)\le\sum_{e\in S}f(e)$. In the remaining cases $B>0$ and $O\ne\emptyset$. We use $O(n\log(1/\epsilon)/\epsilon)$ for the small-$\epsilon$ form of the uniform bound $O((n/\epsilon)\log(2/\epsilon))$, valid for all $0<\epsilon<1$.

	We denote $O$ as an optimal solution with the optimal value $\opt=f(O)$ and $r=\arg\max_{o \in O}c(o)$.
	Another frequently used property of a non-negative submodular function is: For any $T\subseteq V$ and two disjoint subsets $X, Y$ of $V$ we have:
	\begin{align}
	f(T)\leq  f(T\cup X)+ f(T\cup Y).
	\label{ine:sub}
	\end{align}
	We use  this Lemma  to analyze our algorithms' performance.
	\begin{lemma}(Lemma 2.2. in~\cite{Buchbinder14_Submax_Cardinal})
		Let $f: 2^V \mapsto \mathbb{R}_+ $	be submodular. Denote by $A(p)$ a random subset of $A$ where each element appears with probability at most $p$ (not necessary independently). Then $\E[f(A(p))]\geq (1-p)f(\emptyset)$.
		\label{lem:sub2}
	\end{lemma}
	\section{Proposed Algorithms}
	\label{sec:algs}
	In this section, we introduce two main algorithms, $\DLA$ and $\RLA$. The core of these two algorithms lies in our novel design of $\LA$ (Linear Approximation), a $19$-approximation algorithm within $O(n)$ queries. Although its factor approximation is quite large, it is the \textit{first deterministic algorithm} that gives a constant approximation factor within only a linear number of queries for the general $\SMK$ problem. $\LA$ is a key building block for our  $\DLA$ and its randomized version, $\RLA$.
	\subsection{$\LA$ Algorithm}
	The $\LA$  algorithm (Algorithm~\ref{alg:la}) splits the ground set into two subsets $V_1$ and $V_2$. The first contains any element whose cost is at most $B/2$; the second includes the rest.  The key strategy for $\LA$ is dividing the ground set into subsets to quickly find out the bound of the optimal solution in linear queries, then selecting potential elements into two sets to get a constant approximation factor for $\SMK$. 
	
	Since the feasible solution for over $V_2$ contains at most one element, we can bound it by the maximal singleton $e_{max}=\arg\max_{e\in V} f(e)$.  
	For the  subset $V_1$, the algorithm initiates two empty disjoint sets $X, Y$; each has a threshold (ratio of $f$ value over $B$) to consider the admission of a new element. A considered element is added to a set $Z\in \{X, Y\}$ to which it has the higher ratio between  marginal gain and its cost with respect to $Z$ (i.e. ``\textit{density gain}") as long as the density gain is at least $f(Z)/B$. Note that the cost of disjoint sets may be higher than $B$, so we obtain feasible solutions from them by only selecting the last elements added with the cost nearest to $B$ (lines 6-7). Finally, the algorithm returns a feasible solution with the maximum $f$ value.
	
	Note that the approach of \cite{li-linear-knap-nip22} gave a range bound of an optimal solution for the {\bf monotone} $\SMK$ problem in linear time, but it does not work for the {\bf non-monotone } objective function and does not provide any feasible solution. To deal with the non-monotone function, our algorithm  maintains $X$ and $Y$ to be always disjoint and exploit \eqref{ine:sub} to get:
	\begin{align*}
	f(O_1) \leq f(X\cup O_1)+f(Y\cup O_1).
	\end{align*}
	and  bound the optimal value by $f(O)\leq f(O_1)+f(O_2)$
	where $O_1$ and $O_2$ are optimal solutions of the problem over $V_1$ and $V_2$, respectively. 
	
	On the other hand, an advantage of our algorithm is that we can use the $f$ value of the maximal singleton to design and analyze theoretical bounds for our later algorithms. 
	\begin{algorithm}[t]
		\SetNlSty{text}{}{:}
		\KwIn{An instance $(f, V, B)$.} 
		$V_1 \leftarrow \{e \in V: c(e)\leq B/2\}$, $X \leftarrow \emptyset$, $Y \leftarrow \emptyset$ 
		$e_{max} \leftarrow \arg\max_{e\in V}f(e)$ \label{alg:la-finde_max}
		\\
		\ForEach{$e \in V_1$}
		{
			
			Find $Z\in \{X, Y\}$ such that: $Z=\arg\max_{Z \in \{X, Y\}: \frac{f(e|Z)}{c(e)}\geq \frac{f(Z)}{B}}\frac{f(e|Z)}{c(e)}$ \label{alg:la-findZ}
			\\
			\textbf{If} exist such set $Z$ \textbf{then} $Z \leftarrow Z \cup  \{e\}$
		}
		$X' \leftarrow \arg\max_{X(j):  0\leq j\leq |X|, c(X(j))\leq B}c(X(j))$
		\\
		$Y' \leftarrow \arg\max_{Y(j):  0\leq j\leq |Y|, c(Y(j))\leq B}c(Y(j))$, where  $T(j)$ is a set of last $j$ elements added in $T\in \{X, Y\}$.
		\\
		$S \leftarrow \arg\max_{Z \in \{ X', Y', \{e_{max}\} \} }f(Z)$
		\\
		\Return $S$.
		\caption{$\LA$ Algorithm}
		\label{alg:la}
	\end{algorithm}

	Lemma~\ref{lem:la1} provides a bound of 
	optimal solution $V_1$ by  two disjoint sets $X$, $Y$, which is critical to analyze the theoretical bound of Algorithm~\ref{alg:la}. 
	\begin{lemma}
		At the end of the main loop of Algorithm~\ref{alg:la}, we have: 
		$
		f(O_1) \leq 3(f(X)+ f(Y))
		$.
		\label{lem:la1}
	\end{lemma}
	\begin{theorem}
		Algorithm~\ref{alg:la} is deterministic, returns an approximation factor of $19$ and takes $O(n)$ queries.  
		\label{theo:la1}
	\end{theorem}
	\begin{algorithm}[h]
		\SetNlSty{text}{}{:}
		\KwIn{An instance $(f, V, B)$, parameters  $p, \alpha$} 
		$e_{max}\leftarrow \arg\max_{e \in V}f(e)$,
		$V_1 \leftarrow \{e \in V| c(e)\leq B/2\}$
		\\
		$V_p\leftarrow \{e \in V_1: \mbox{Select $e$ with probability $p$}\}$, $S \leftarrow \emptyset$
		\\
		\ForEach{$e \in V_p$}
		{
			
			\textbf{If} $f(e|S)/c(e)\geq \alpha f(S)/B$ \textbf{then}	$S\leftarrow S \cup \{e\}$
			
		}
		$S' \leftarrow \arg\max_{S(j):  0\leq j\leq |S|, c(S(j))\leq B}c(S(j))$, where $S(j)$ is a set of last $j$ elements added into $S$.
		\\
		$S \leftarrow \arg\max_{T \in \{ S', \{e_{max}\} \} }f(T)$
		\\
		\Return $S$.
		\caption{$\LAR$ Algorithm}
		\label{alg:laran}
	\end{algorithm}
	We further introduce the $\LAR$ (Algorithm~\ref{alg:laran}) algorithm, a randomized version of Algorithm~\ref{alg:la}. $\LAR$ selects $V_p$ from $V_1$ by selecting each $e \in V_1$ independently with probability $p$ and scanning the sampled items in the fixed order, then it builds the candidate set $S$ from $V_p$ instead of maintaining two disjoint sets.
	Although $\LAR$ is a randomized algorithm, it provides a better approximation factor of $\LA$ and be used for designing a later randomized algorithm $\RLA$.
	\begin{theorem}
		Algorithm~\ref{alg:laran} takes $O(n)$ queries and returns an approximation factor of $16.034$ in expectation with $p=\sqrt{2}-1$ and $\alpha=\sqrt{2+2\sqrt{2}}$.
		\label{theo:lar}
	\end{theorem}
	Due to space limit, proofs of Lemmas, Theorems~ \ref{theo:la1} and ~\ref{theo:lar} are provided in the Appendix.
	\subsection{$\DLA$ Algorithm}
	We now introduce our $\DLA$ (Algorithm~\ref{alg:dla}), a \textbf{D}eterministic and \textbf{L}inear query complexity \textbf{A}pproximation algorithm that has an approximation factor of $6+\epsilon$. 
	The key strategy is combining the properties of two disjoint sets with a greedy threshold to construct several candidate solutions to analyze the theory of the non-monotone objective function. 
	
	$\DLA$ takes an instance $(f, V, B)$ and a parameter $\epsilon$ as inputs. $\DLA$ consists of two phases. At the first one (lines~\ref{dla:p1-b}-\ref{dla:p1-e}), the algorithm calls $\LA$ as a subroutine to obtain a candidate solution $S'$ and get an approximate range of optimal value $[\Gamma, 19\Gamma]$ where $\Gamma=f(S')$ (line 1). It then adapts the greedy threshold to add elements with high-density gain into two disjoint sets $X$ and $Y$. Specifically, this phase consists of multiple iterations; each scans one time over the ground set (lines 3-9). An element added to the set $T \in \{X, Y\}$ to which has the higher density gain without violating the budget constraint, as long as the density gain is at least $\theta$, which initiates to $19\Gamma/(6\epsilon'B)$ and decreases by a factor of $(1-\epsilon')$ after each iteration until less than to $\Gamma(1-\epsilon')/(6B)$,  where $\epsilon'=\epsilon/14$. 
	
	The second phase (lines~\ref{dla:p2-b}-\ref{dla:p2-e}) is to improve the quality of candidate solution $T\in \{X, Y\}$ which was obtained at the end of phase 1.  Denote  $T^i$ as a set of the first $i$ elements (with $T^0=\emptyset$) added in $T$ in phase 1. 	
	Our main observation is that the performance of $\DLA$ depends on the cost of $T'=\arg\max_{T^i: c(T^i)\leq B-c(r)}(i)$. Recall that $r$ is $\arg\max_{o\in O}c(o)$ and $c(r)\leq B$.
	We scan an upper bound of $c(T')$ from $\epsilon'B$ to $B$ and improve the quality of $T'$ by adding into it an element $e=\arg\max_{e\in V: c(T'\cup  \{e\})\leq B}f(T'\cup \{e\})$ (lines 13-15).
	\paragraph{}
	The following Lemmas give the bounds of the final solution when $c(r)< (1-\epsilon')B$ and $c(r)\geq(1-\epsilon')B$, respectively.
	\begin{algorithm}
		\SetNlSty{text}{}{:}
		\KwIn{An instance $(f, V, B)$, $0<\epsilon<1$ } 
		$S'\leftarrow \LA(f, V, B),\Gamma \leftarrow f(S'), \epsilon' \leftarrow \frac{\epsilon}{14}$ \label{dla:p1-b}
		\\
		$ \Delta \leftarrow \lceil\frac{ \log(1/\epsilon')}{\log(1+\epsilon')} \rceil, \theta \leftarrow 19\Gamma/(6\epsilon' B), X\leftarrow \emptyset, Y \leftarrow \emptyset$
		\\
		\While{$\theta \geq \Gamma(1-\epsilon')/(6B)$}
		{
			\ForEach{$e \in V \setminus (X\cup Y)$}
			{
                Among $T\in\{X,Y\}$ with $c(T)+c(e)\le B$ and $f(e\mid T)\ge\theta c(e)$, choose one maximizing $f(e\mid T)/c(e)$.
				\\
				\textbf{ If} exist such set $T$ \textbf{then} $T \leftarrow T \cup \{e\}$
				
			}
			$\theta \leftarrow (1-\epsilon')\theta$\\
		}
		\label{dla:p1-e}
		\For{$l=0$ to $ \Delta$}
		{ \label{dla:p2-b}
			$X'_{(l)} \leftarrow \arg\max_{X^i: c(X^i)\leq \epsilon'B(1+\epsilon')^l, i\leq |X|}i$
			\\
			$Y'_{(l)} \leftarrow \arg\max_{Y^i: c(Y^i)\leq \epsilon'B(1+\epsilon')^l, i\leq |Y|}i$
			\\
			$e_X \leftarrow  \arg\max_{e \in V:  c(X'_{(l)}\cup \{e\})\leq B}f(X'_{(l)}\cup \{e\})$ \label{alg:dla-select_ex}
			\\
			$e_Y \leftarrow  \arg\max_{e \in V:  c(Y'_{(l)}\cup \{e\})\leq B}f(Y'_{(l)}\cup \{e\})$
			\\
			$X_{(l)}\leftarrow X'_{(l)} \cup \{e_X\}$, 
			$Y_{(l)}\leftarrow Y'_{(l)} \cup \{e_Y\}$
		} \label{dla:p2-e}
		$S \leftarrow \arg \max_{T \in \{S', X, Y, X_{(0)},.., X_{( \Delta )}, Y_{(0)}, .., Y_{( \Delta )} \}}f(T)$
		\\ 
		\Return $S$.
		\caption{$\DLA$ Algorithm}
		\label{alg:dla}
	\end{algorithm}
	\begin{lemma}
		If $c(r)< (1-\epsilon')B$, one of two things happens:
		\textbf{a)} $f(S)\geq\frac{1}{6(1+\epsilon')}\opt$; \textbf{b)} There exists a subset $X' \subseteq X$ so that $f(O\cup X')\leq 2f(S)+ \max\{\frac{1+\epsilon'}{1-\epsilon'}f(S), \frac{(1-\epsilon')}{6}\opt\}$. 
		\\
		Similarly, one of two conditions happens:
		\textbf{c)} $f(S)\geq\frac{1}{6(1+\epsilon')}\opt$; \textbf{d)} There exists a subset $Y' \subseteq Y$ so that $f(O\cup Y')\leq 2f(S)+ \max\{\frac{1+\epsilon'}{1-\epsilon'}f(S),\frac{(1-\epsilon')}{6}\opt\}$.
		\label{lem:dla1}
	\end{lemma}
	\begin{lemma}
		If $c(r)\geq (1-\epsilon')B$, one of two things happens:
		\textbf{e)} $f(S)\geq\frac{(1-\epsilon')^2}{6}\opt$; \textbf{f)} There exists a subset $X' \subseteq X$ so that $f(O\cup X')\leq 2f(S)+ \max\{\frac{(1-\epsilon')\opt}{6},f(S)+\frac{\epsilon'\opt}{6}\}$.
		\\
		Similarly,  one of two things happens:
		\textbf{g)} $f(S)\geq\frac{(1-\epsilon')^2}{6}\opt$; \textbf{h)} There exists a subset $Y' \subseteq Y$ so that $f(O\cup Y')\leq 2f(S)+ \max\{\frac{(1-\epsilon')\opt}{6},f(S)+\frac{\epsilon'\opt}{6}\}$. 
		\label{lem:dla2}
	\end{lemma}
	\begin{theorem}
		For any $\epsilon \in (0, 1)$,	$\DLA$ is a deterministic algorithm that has a query complexity $O(n\log(1/\epsilon)/\epsilon)$ and returns an approximation factor of $6+ \epsilon$.
		\label{theo:dla}
	\end{theorem}
	\begin{proof}
	Write $d=\epsilon'=\epsilon/14$ and $M=f(S)$. Suppose first that $c(r)<(1-d)B$. If conclusion \textbf{a)} or \textbf{c)} of Lemma~\ref{lem:dla1} holds, then $\opt\le 6(1+d)M$. Otherwise its two prefixes are disjoint. Applying~\eqref{ine:sub} gives
	\begin{equation}\label{ine:theo1.2}
	\opt\le 4M+2\max\left\{\frac{1+d}{1-d}M,\frac{1-d}{6}\opt\right\}.
	\end{equation}
	The first term in the maximum gives factor $6+4d/(1-d)$. The second gives factor $12/(2+d)\le6$.

	If $c(r)\ge(1-d)B$ and conclusion \textbf{e)} or \textbf{g)} of Lemma~\ref{lem:dla2} holds, then $\opt\le6M/(1-d)^2$. Otherwise conclusions \textbf{f)} and \textbf{h)} and~\eqref{ine:sub} give
	\begin{equation}\label{ine:theo1.1}
	\opt\le 4M+2\max\left\{\frac{1-d}{6}\opt,M+\frac d6\opt\right\}.
	\end{equation}
	The two terms give factors $12/(2+d)$ and $6/(1-d/3)$, respectively. All these factors are at most $6+14d=6+\epsilon$. In particular,
	$[6/(1-d)^2-6]/d=6(2-d)/(1-d)^2\le12(14/13)^2<14$.

	There are at most $1+\log(19/[d(1-d)])/[-\log(1-d)]$ density scans and $1+\lceil\log(1/d)/\log(1+d)\rceil$ prefix caps. Each scan and each pair of augmentations costs at most $2n$ queries with cached chain values. Together with Algorithm~\ref{alg:la}, this gives $O((n/\epsilon)\log(2/\epsilon))$ queries. The returned candidates are all feasible.
	\end{proof}
	\subsection{$\RLA$ Algorithm}
	We further introduce the $\RLA$ (Algorithm~\ref{alg:rla}), a \textbf{R}andomized and \textbf{L}inear query complexity \textbf{A}pproximation algorithm with factor $4+\epsilon$ in expectation. It uses the same decreasing-density scan and single-item prefix augmentation as $\DLA$. The solution returned by $\LAR$ remains a candidate. A call to $\LA$ supplies the deterministic upper bound $U=19\Gamma\ge\opt$ needed to set the first threshold.

	An \emph{offer} is the first time an item passes both the budget and density tests. Each offer receives exactly one independent fair coin. Let $S_i$ be the accepted set after offer $i$, with $S_0=\emptyset$. For a cap $q$, let $P_q$ be the last accepted prefix of cost at most $q$. Also let $C_q$ be the prefix immediately before the first offer $u_i$ with $c(S_{i-1})+c(u_i)>q$. If no offer crosses $q$, set $C_q=S_j$, the final chain. Whether an offer crosses is known before its coin is tossed. Both prefixes cost at most $q$. Define $\operatorname{Aug}(P)$ as the highest-value set among $P$ and all feasible sets $P\cup\{e\}$, $e\in V$.
	\begin{algorithm}[t]
		\SetNlSty{text}{}{:}
		\KwIn{An instance $(f,V,B)$, $0<\epsilon<1$.}
		$S'\leftarrow\LAR(f,V,B,\sqrt2-1,\sqrt{2+2\sqrt2})$, $D_0\leftarrow\LA(f,V,B)$\;
		$\Gamma\leftarrow f(D_0)$, $U\leftarrow19\Gamma$, $d\leftarrow\epsilon/10$, $a\leftarrow\epsilon/100$\;
		$S_0\leftarrow\emptyset$, $j\leftarrow0$, $\theta\leftarrow 2U/(aB)$\;
		$S^*\leftarrow\arg\max_{T\in\{D_0,S'\}} f(T)$\;
		\While{$\theta\ge(1-d)d\Gamma/B$}{
			\ForEach{$e\in V\setminus\{u_1,\ldots,u_j\}$ in fixed order}{
				\If{$c(S_j)+c(e)\le B$ and $f(e\mid S_j)\ge\theta c(e)$}{
					$u_{j+1}\leftarrow e$; record $(S_j,e)$\;
					Keep the better of $S^*$ and $S_j\cup\{e\}$ in $S^*$\;
					With probability $1/2$ set $S_{j+1}\leftarrow S_j\cup\{e\}$; otherwise set $S_{j+1}\leftarrow S_j$\;
					$j\leftarrow j+1$\;
				}
			}
			$\theta\leftarrow(1-d)\theta$\;
		}
		\For{$l=0$ to $\lceil\log(1/a)/\log(1+d)\rceil$}{\label{rla:p2-b}
			$q\leftarrow\min\{B,aB(1+d)^l\}$\;
			Keep the best of $S^*,\operatorname{Aug}(P_q),\operatorname{Aug}(C_q)$ in $S^*$\;\label{rla:p2-e}
		}
		\If{some $v\in V$ has $c(v)>(1-a)B$}{
			Keep the better of $S^*$ and Algorithm~\ref{alg:cheap}$(f,V,B,U,a)$ in $S^*$\;
		}
		\Return $\arg\max_{T\in\{S^*,S_j\}} f(T)$\;
		\caption{$\RLA$ Algorithm}\label{alg:rla}
	\end{algorithm}
	When a nearly full-budget item exists in the input, Algorithm~\ref{alg:cheap} in Appendix~\ref{app:ran-proof} computes one additional candidate using $O(n)$ queries. It first builds a feasible union of inexpensive items and then applies randomized double greedy~\cite{usm} to that union. This step is invoked at most once, regardless of the number of expensive items.

	\begin{lemma}\label{lem:rla1}
	Fix a cap $q$ and a set $H\subseteq V$ independently of the scan coins. For each offer $e$ strictly before the first crossing let $w_e$ be its gain at the time of the offer. Count all offers if no crossing occurs. Let $Z_H$ be the sum of $w_e$ over rejected offers in $H$ and accepted offers outside $H$, up to this stop. Then
	\[
	\E[Z_H]=\E[f(C_q)],\qquad \E[f(O\cup C_q)]\ge\opt/2.
	\]
	\end{lemma}
	\begin{theorem}
		For any $\epsilon\in(0,1)$, Algorithm~\ref{alg:rla} returns a feasible $S$ with $\E[f(S)]\ge\opt/(4+\epsilon)$. Its worst-case query complexity is $O((n/\epsilon)\log(2/\epsilon))$.
		\label{theo:rla}
	\end{theorem}
	The proof of Lemma~\ref{lem:rla1} and Theorem~\ref{theo:rla} appears in Appendix~\ref{app:ran-proof}.
	\section{Experimental Evaluation}
	\label{sec:expr}
	The measurements in this section and the additional experimental results in the appendix are from the original implementations. The revised pseudocode in Algorithms~\ref{alg:dla} and~\ref{alg:rla} has not been rerun on these datasets.
	 In this section, we compare the performance between our algorithms and state-of-the-art algorithms for the $\SMK$ problem on three  applications: Revenue Maximization, Image Summarization, and Maximum  Weighted Cut.
	\subsection{Applications And Datasets}
	\paragraph{Revenue Maximization.}Given a social network that represented by a graph $G=(V, E)$ where $V$ and  represent a set of users a set of user connections, respectively Each edge $(u, v)$ assigned a weight $w_{(u, v)}$ that reflects  the ``closeness'' of $u$ and $v$. We follow  \cite{fast_icml} to define the advertising revenue of any node set $S\subseteq V$ as  
	$
	f(S)=\sum_{u \in V\setminus S}\sqrt{\sum_{v \in S: (v, u)\in E} w_{(u, v)}}
	$.
	The weight $w_{(u,v)}$  is randomly sampled from the
	continuous uniform distribution $U(0, 1)$ as in and each node $u$ has a cost $c(u)=g(\sqrt{\sum_{(u,v)\in E}w_{(u,v)}})$ where $g(x)=1-e^{-\mu x}$ and $ \mu= 0.2$ \cite{Han2021_knap}.  Given a budget of $B$, the goal of the problem is to select a set $S$ with the cost at most $B$ to maximize $f(S)$. This problem is an instance of non-monotone $\SMK$ \cite{Han2021_knap}. In this application, we utilized the ego-Facebook dataset from  \cite{leskovec_celf} which consists of over $4$K nodes and over $88$K edges.
	\paragraph{Image Summarization.}
	Given a graph $G$ = ($V$, $E$) where each node $u\in V$ represents an image, and each edge $(u,v)\in E$ is assigned a weight $w_{u,v}$ representing the similarity between image $u$ and image $v$. Define $c(u)$ the cost to collect the image $u$. The goal is to identify a representative subset $S\subseteq V$ with a limited budget $B$ that maximizes the representative value defined as $
	f(S) = \sum_{u\in V} \max_{v\in S} w_{u,v} - \frac{1}{|V|} \sum_{u\in V}\sum_{v\in S} w_{u,v}$ \cite{fast_icml,Han2021_knap}. The function $f(\cdot)$ is non-monotone, non-negative, and submodular \cite{fast_icml}. Following the recent work \cite{Han2021_knap,fast_icml}, we set this instance as follows: We first randomly selected $500$ images from the CIFAR data sets \cite{img-data,fast_icml}, which contained $10.000$ images. We then measure the similarity between image $u$ and image $v$ by using the cosine similarity of their $3.072$-dimensional pixel vectors. Finally, we use Root Mean Square (RMS) contrast as a metric to evaluate the quality of the images and assign a cost to each image based on its RMS contrast.
	\paragraph{Maximum Weighted Cut.}
	Given a graph $G = (V, E)$, and a non-negative edge weight $w_{u,v}$ for each $(u,v) \in E$. For a set of nodes $S\subset V$, define the weighted cut function $f(S) = \sum_{u\in V \setminus S} \sum_{v \in S} w_{u,v}$. The maximum (weighted) cut problem is to find a subset $S\subseteq V$ such that the  $f(S)$ is maximized. It is indicated in \cite{Kuhnle19-nip,Amanatidis_sampleGreedy} as $f(\cdot)$ is non-monotone and submodular. The datasets used in the application included an Erdős-Rényi (ER) random graph with 5000 nodes, and an edge probability of $0.2$  and the cost of each node $c(u)$ was randomly uniformly chosen from the range $(0,1)$ as in \cite{Amanatidis_sampleGreedy}.
	\paragraph{Experiment Settings.}
	We compare our algorithms with the applicable  state-of-the-art algorithms listed below:
	\begin{itemize}
		\item \textbf{FANTOM}: The randomized algorithm of \cite{fast_icml} with the expected factor of the $10+\epsilon$ in $O(n^2 \log(n)/\epsilon)$.
		\item\textbf{SAMPLE GREEDY}: The  randomized algorithm of \cite{Amanatidis_sampleGreedy} with the  factor $5.83+\epsilon$ in  query complexity of $O(n\log(n/\epsilon)/\epsilon)$ queries. For the easy following, we refer to SAMPLE GREEDY as the \textbf{GREEDY}.
		\item\textbf{SMKDETACC}: The deterministic algorithm of \cite{Han2021_knap} with the  factor $6+\epsilon$ in $O(n\log(k/\epsilon)/\epsilon)$ queries. This is the fastest deterministic approximation algorithm for $\NSMK$.
		\item\textbf{SMKRANACC}: This is the fastest randomized algorithm of  \cite{Han2021_knap} with the expected approximation factor $4+\epsilon$ in  query complexity of $O(n\log(k/\epsilon)/\epsilon)$.
		\item \textbf{SMKSTREAM}: The first streaming algorithm   for studied problem that returns the approximation factor of $6+\epsilon$ within $O(n\log(B)/\epsilon)$ queries \cite{Han2021_knap}.
	\end{itemize}

	In our experiments, the budget range from 2\% to 12\% of the total cost of the ground sets as setting of \cite{Amanatidis_sampleGreedy}. We set  $\epsilon=0.1$  for all algorithms and $\alpha=\beta=1/6,h=2,r=2$ for SMKSTREAM \cite{Han2021_knap}. 
	\subsection{Experiment Results}
	\begin{figure}[h]
		\centering
		{\includegraphics[width=.495\linewidth]{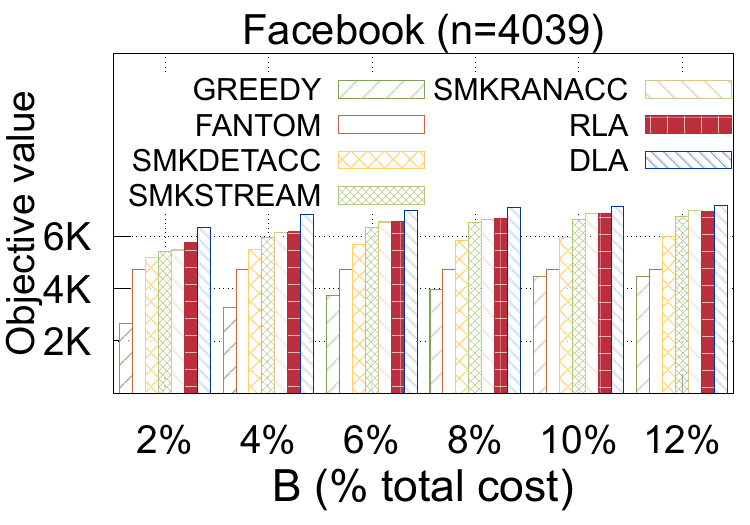}}
		{\includegraphics[width=.495\linewidth]{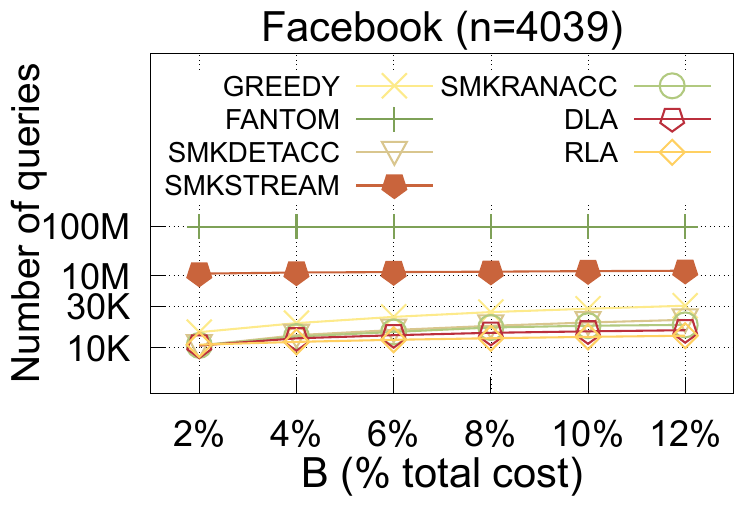}}
		\\ 
		\footnotesize{ \hspace{.5cm}	(a) \hspace{4cm} (b)} 
		\\
		{\includegraphics[width=.495\linewidth]{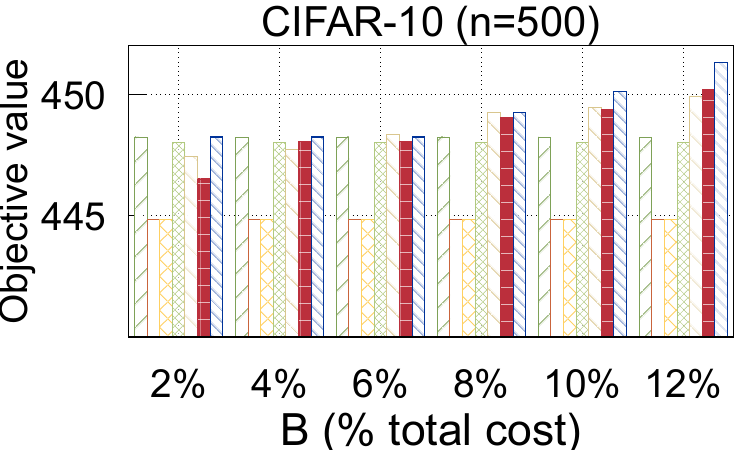}}
		{\includegraphics[width=.495\linewidth]{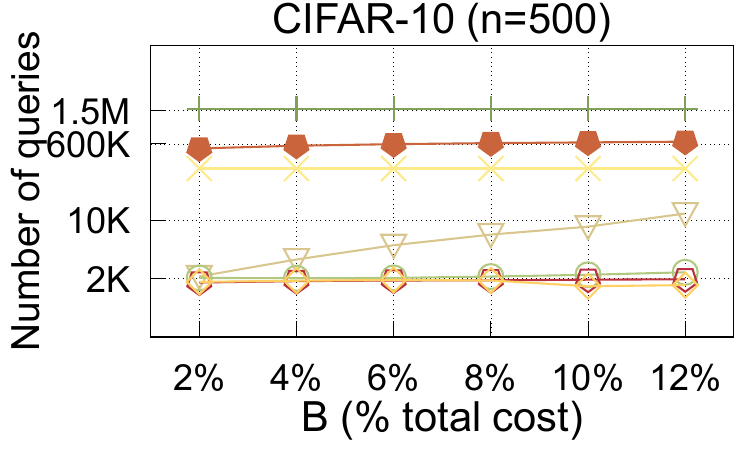}}
		\\
		\footnotesize{ \hspace{.5cm}	(c) \hspace{4cm}(d) }
		\\
		{\includegraphics[width=.495\linewidth]{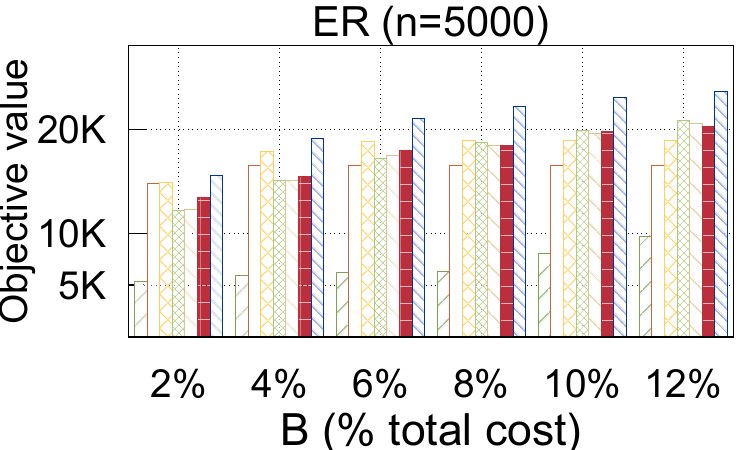}}
		{\includegraphics[width=.495\linewidth]{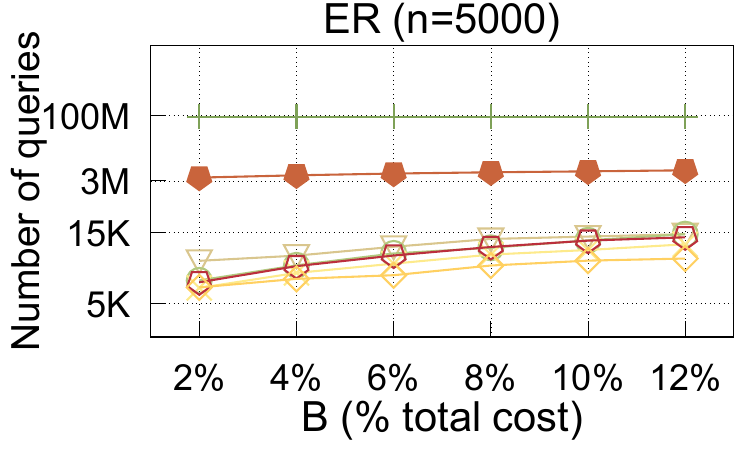}}
		\\
		\footnotesize{	 \hspace{.5cm} (e) \hspace{4cm}(f)}
		\caption{Performance of algorithms for $\NSMK$ on three instances: (a), (b) Revenue Maximization; (c), (d) Image Summarization and (e), (f)  Maximum Weighted Cut.}
		\label{fig:res}
	\end{figure}
	The result of the experiment is shown in Figure \ref{fig:res}. 
	First, Figures \ref{fig:res}(a)(c)(e) represent the quality of algorithms via values of the objective function on $3$ instances.  As can be seen, $\DLA$ always gives the highest values at all $B$ milestones in all instances. $\RLA$, SMKRANACC, and SMKSTREAM are not much different. FANTOM results lower while GREEDY provides the lowest. 
	Regarding the deterministic algorithm, $\DLA$ is several tens to thousands of units better than SMKDETACC on (a) and (e), especially, several times higher on (c). Regarding the randomized algorithm, our $\RLA$ gives as well quality as SMKRANACC. This result insists that our algorithm ensures good performance compared to the algorithms of \cite{Han2021_knap}, which are currently the best. In the end, our algorithm tends to be considerably better than the rest when $B$'s values increase.
	
	Figures \ref{fig:res}(b)(d)(f) illustrate the number of queries of the above algorithms. FANTOM is the highest, SMKSTREAM is the second, and the remaining is much lower. FANTOM and SMKSTREAM require millions of queries, whereas the rest is thousands of times lower than them. In the rest, the GREEDY's queries are also usually higher than the others except on (f). Queries of $\DLA$, $\RLA$, and SMKRANACC look similar while queries of SMKDETACC change due to different datasets. $\RLA$ spends the fewest queries, and  $\DLA$ needs fewer queries than that of SMKDETACC.
	When $B$ grows, the number of queries of $\RLA$ increases slowest, whereas the queries of SMKDETACC increase fastest. Especially, in  Image Summarization, $\DLA$, $\RLA$, SMKDETACC, and SMKRANACC are all approximately $2K$ at $B=2\%$ the total cost; however, SMKDETACC is $5$ times higher than the rest when $B=12\%$.

	On the whole, our algorithms, $\DLA$, and $\RLA$ keep the balance between performance guarantee and query complexity. It's extremely important to save running time with big data. Moreover, experimental results show that our algorithms are efficient ones comparable to state-of-the-art algorithms.
	\section{Conclusions}
	\label{sec:con}
	Motivated by the challenge of solving the $\NSMK$ on the massive data, in this work, we proposed two approximation algorithms $\DLA$, $\RLA$. 
	To the best of our knowledge, our algorithms are the first to achieve a constant factor approximation for the considered problem in linear query complexity. 
	The original implementations' performance in solution quality and computation complexity compared to state-of-the-art algorithms was supported by the experiment results in three real-world applications. 
	\section*{Acknowledgements}
	This work was supported in part by the National Science Foundation (NSF) grants IIS-1908594.
	\bibliographystyle{named}
	\bibliography{SMK-ref}
	\newpage
	\clearpage
	\appendix
	\onecolumn
	\section*{Appendix}	
	\section{Proofs of $\LA$ and $\LAR$ Algorithms}
	In this section, we use the following Lemma to analyze our algorithms' performance.
	\begin{lemma}[Lemma 2.2 of \cite{Feige_nonmono_submax}]
		Let $f: 2^V \mapsto \mathbb{R}^+ $ submodular. Denote by $A(p)$ a random subset of $A$ where each element appears with probability exactly $p$ (without requiring independence). Then $\E[f(A(p))]\geq (1-p)f(\emptyset) + p\cdot f(A)$.
		\label{lem:sub1}
	\end{lemma}
	\begin{proof}Order $A=\{a_1,\ldots,a_m\}$ and let $A_i=\{a_1,\ldots,a_i\}$. By diminishing returns, $f(R)\ge f(\emptyset)+\sum_{a_i\in R}(f(A_i)-f(A_{i-1}))$ for every $R\subseteq A$. Taking expectations for $R=A(p)$ gives the claim by telescoping. For the at-most-$p$ version in Lemma~\ref{lem:sub2}, order the items instead by their nonincreasing inclusion probabilities $p_1\ge\cdots\ge p_m$, and put $p_{m+1}=0$. The same expansion gives the lower bound $(1-p_1)f(\emptyset)+\sum_{i=1}^m(p_i-p_{i+1})f(A_i)\ge(1-p)f(\emptyset)$ by nonnegativity. The empty ground set is immediate.\end{proof}
	Besides, we use the following notation: For  $e\in X \cup Y$, we denote by $X^{<e}$ and $Y^{<e}$ the set of elements in $X$ and $Y$ before adding $e$ into  $X$ or $Y$, respectively.
	\begin{proof}[Proof of Lemma~\ref{lem:la1}]
		Fix the final chain $X$. For $e\in O_1\setminus X$, write $X^{<e},Y^{<e}$ for the chain states when $e$ was scanned. Partition these items according to whether they passed the threshold of $X$ at that time. For an item that failed, diminishing returns and the nonnegative accepted gains in $X$ imply
		\[ f(e\mid X)\le f(e\mid X^{<e})<c(e) f(X^{<e})/B\le c(e)f(X)/B.\]
		This bound includes items that failed the test for $X$ and were subsequently routed to $Y$ in the same iteration. An item that passed the test for $X$ but was not assigned there must have joined $Y$. As both were eligible choices, its gain in $Y$ satisfies $f(e\mid X)\le f(e\mid X^{<e})\le f(e\mid Y^{<e})$. Summing the first class costs at most $f(X)c(O_1)/B\le f(X)$. The second class uses distinct nonnegative gains in $Y$, whose total is at most $f(Y)$. Hence
		\[ f(O_1\cup X)\le f(X)+\sum_{e\in O_1\setminus X}f(e\mid X)\le 2f(X)+f(Y). \]
		Interchanging the chains gives $f(O_1\cup Y)\le 2f(Y)+f(X)$. Inequality~\eqref{ine:sub} applies because $X\cap Y=\emptyset$, giving $f(O_1)\le3(f(X)+f(Y))$.
	\end{proof}
	\begin{proof}[Proof of Theorem \ref{theo:la1}]
        The singleton scan costs $n$ queries. If $V_1$ is empty, both suffix values are known. Otherwise the first light item's two marginals equal its cached singleton. At most $2(|V_1|-1)$ further queries compute all remaining marginal gains, using cached chain values, and at most two more evaluate the suffixes. The total is at most $n+2|V_1|\le 3n$.
		We now prove the approximation factor of the algorithm. Denote  $X=\{x_1, x_2, \ldots, x_{|X|}\},  X^i=\{x_1, x_2, \ldots, x_i\},  1\leq i\leq |X|, X^0=\emptyset$. Suppose that $X_1=X\setminus X'=\{x_1, \ldots, x_l\}$ and  $X'=\{x_{l+1}, x_{l+2}, \ldots, x_{|X|}\}$. 
		
		We first show a claim that $f(X_1)\leq 2f(X)/3$. If $c(X)<B$, $X'=X$ and $X_1=\emptyset$, the claim is true. Therefore we consider the case $c(X)\geq B$.
		The selection rule for adding each element gives
		$$
		f(X^j)=f(X^{j-1})+f(x_j|X^{j-1})\geq f(X^{j-1})+\frac{c(x_j)}{B}f(X^{j-1}) \geq f(X^{j-1})
		$$
		for $1 \leq j \leq |X|$. Therefore: 
		\begin{align*}
		f(X)-f(X_1)&=\sum_{i=1}^{|X'|}f(x_{l+i}|X_1\cup \{x_{l+1}, \ldots ,x_{l+i-1}\})
		\\
		&\geq \sum_{i=1}^{|X'|}\frac{c(x_{l+i})}{B}f(X_1 \cup \{x_{l+1}, \ldots ,x_{l+i-1}\})
		\\
		& \geq \sum_{i=1}^{|X'|}\frac{c(x_{l+i})}{B}f(X_1)  =\frac{c(X')}{B}f(X_1)
		\\ 
		& \geq \frac{f(X_1)}{2}.
		\end{align*}
		The last inequality is because each element in $V_1$ has the cost at most $B/2$, and the selection rule of $X'$ we have 
		$$B \geq c(X')\geq B-B/2=B/2.$$
		Thus, $f(X_1)\leq 2f(X)/3$. By the submodularity of $f$, we obtain: 
		\begin{align*}
		f(X')\geq f(X)-f(X_1)\geq \frac{f(X)}{3}.
		\end{align*}
		Similarly, we  have: $f(Y') \geq f(Y)/3$. On the other hand, the optimal solution over $V_2$ has at most one element, so $f(O_2)\leq  f(e_{max})$. Finally, from the submodularity, the selection rule of the final solution of the Algorithm \ref{alg:la} and Lemma \ref{lem:la1} we obtain:
		\begin{align*}
		f(O)& \leq f(O \cap V_1)+f(O \cap V_2) 
		\\
		&\leq f(O_1)+f(O_2)
		\\
		& \leq 3 (f(X)+f(Y))+ f(e_{max})
		\\
		& \leq 9 (f(X')+f(Y'))+ f(e_{max}) \leq 19 f(S)
		\end{align*}
		which completes the proof.
	\end{proof}
	Before proving Theorem \ref{theo:lar}, we provide the following suffix bound. Here $S$ is the chain at the end of the scan, before its variable is reused for the returned solution, and $S'$ is its longest feasible suffix.
	\begin{lemma}
		$f(S')\geq \frac{\alpha}{\alpha+2}f(S)$.
		\label{lem:lar}
	\end{lemma}
	\begin{proof}
		If $S'=S$, the Lemma holds. We consider the case when $S' \subset S$. Denote $S_1=S\setminus S'$. Denote  $S=\{s_1, s_2, \ldots, s_{|S|}\},  S^i=\{s_1, s_2, \ldots, s_i\},  1\leq i\leq |S|, S^0=\emptyset$. Supposing that $S_1=S\setminus S'=\{s_1, \ldots, s_l\}$ and  $S'=\{s_{l+1}, s_{l+2}, \ldots, s_{|S|}\}$. By the selection rule of each $s_i \in S$, it's easy to see that $f(S^{i+1})\ge f(S^i)$. By the similar argument with the proof of Theorem \ref{theo:la1}, we have
		\begin{align*}
		f(S)-f(S_1)&=\sum_{i=1}^{|S'|}f(s_{l+i}|S_1\cup \{s_{l+1}, \ldots ,s_{l+i-1}\})
		\\
		&\geq \sum_{i=1}^{|S'|}\frac{c(s_{l+i})\alpha}{B}f(S_1 \cup \{s_{l+1}, \ldots ,s_{l+i-1}\})
		\\
		& \geq \sum_{i=1}^{|S'|}\frac{c(s_{l+i})\alpha}{B}f(S_1)  =\frac{\alpha c(S')}{B}f(S_1)
		\\ 
		& \geq \frac{\alpha f(S_1)}{2}
		\end{align*}
		which implies $f(S_1)\leq \frac{2}{2+\alpha}f(S)$. By the submodularity of $f$, we obtain:
		\begin{align*}
		f(S')\geq f(S)-f(S_1)\geq \frac{\alpha}{\alpha+2}f(S).
		\end{align*}
		This completes the proof.
	\end{proof}
	\begin{proof}[Proof of Theorem~\ref{theo:lar}]
		Let $A$ denote the chain before suffix extraction and let $M_R$ denote the value returned. Fix $H\subseteq V_1$ with $c(H)\le bB$. Couple the prescribed independent sampling of $V_p$ with one Bernoulli$(p)$ coin per item of $V_1$. Scan all items in the fixed order, test eligibility at the current chain and reveal an eligible item's coin after that test. The state update is exactly the same as in Algorithm~\ref{alg:laran} for the same vector of coins. A failed test leaves the chain unchanged for either outcome.

		By Lemma~\ref{lem:sub2} applied to $g(T)=f(H\cup T)$, $\E[f(H\cup A)]\ge(1-p)f(H)$ because each item enters $A$ with probability at most $p$. Consider $e\in H\setminus A$. If $e$ failed its test, $f(e\mid A)\le\alpha c(e)f(A)/B$. Otherwise it was eligible and its coin was tails. Let $R$ be the sum of the gains at the time of the test over all such rejected items in $H$. Submodularity gives $f(H\cup A)\le(1+\alpha b)f(A)+R$. Conditional on the history before any eligible coin, its gain is fixed. In $H$ its expected rejected gain is $(1-p)/p$ times its expected accepted gain. The latter gains sum to at most $f(A)$. Therefore
		\begin{equation}\label{eq:lar-fixed-comparator}
		 (1-p)f(H)\le(p^{-1}+\alpha b)\E[f(A)].
		\end{equation}
		Write $h=1+2/\alpha$. Lemma~\ref{lem:lar} gives $f(A)\le hM_R$ on every run. If $O\subseteq V_1$, use $H=O$ and $b=1$. Otherwise $H=O\cap V_1$ costs less than $B/2$ and the sole remaining optimum item has singleton value at most $M_R$. Apply~\eqref{eq:lar-fixed-comparator} with $b=1/2$ and use $f(O)\le f(H)+f(O\setminus H)$. In the two cases the factors multiplying $\E[M_R]$ are respectively
		\[ h\frac{p^{-1}+\alpha}{1-p}<15.038,
		\qquad 1+h\frac{p^{-1}+\alpha/2}{1-p}<12.456 \]
		at $p=\sqrt2-1$ and $\alpha=\sqrt{2+2\sqrt2}$. Both give the stated factor $16.034$. The singleton scan costs at most $n$ queries. If $V_p$ is nonempty, the first marginal is a cached singleton, leaving at most $|V_p|-1$ new marginal queries and one suffix-value query. If $V_p$ is empty the suffix value is already known. Thus at most $n+|V_1|\le2n$ queries are used.
	\end{proof}
	\section{Proofs of $\DLA$ Algorithm}\label{app:det-proof}
	Write $d=\epsilon'$, $L=B-c(r)$, $K=O\setminus\{r\}$, $M=f(S)$ and $\kappa=(1+d)/(1-d)$. In Algorithm~\ref{alg:dla}, $\Gamma=f(S')\le M$ and $\Gamma\le\opt\le19\Gamma$. Every accepted gain is nonnegative. Thus every prefix has value at most the final value of its chain, which is at most $M$. By the original geometric stopping rule the last scanned density satisfies
	\begin{equation}\label{eq:dla-last}
	(1-d)\Gamma/(6B)\le\theta_{\rm last}<\Gamma/(6B).
	\end{equation}
	The strict upper bound follows because the next threshold is below the loop cutoff.

	\begin{proof}[Proof of Lemma~\ref{lem:dla1}]
	Here $L>dB$. The prefix caps start at $dB$, increase by $1+d$, and reach at least $B$. Choose the largest cap $q\le L$. Then $q\le L<(1+d)q$. Consider either chain $T\in\{X,Y\}$, let $T'$ denote the opposite chain, and let $P$ be the longest prefix of $T$ of cost at most $q$. If $T$ crosses $q$, take its snapshot immediately before the first append $u$ causing that crossing. Otherwise take the snapshot at the end of the entire scan phase. These choices include all intervening appends to $T'$.

	Every optimum item is feasible for $T$ up to the snapshot since its cost is at most $c(r)$ and $c(P)\le L$. If $e\in O\setminus P$ was assigned to $T'$ before the snapshot, its marginal at $P$ is at most its actual gain on $T'$. Indeed, if $T$ passed the density test, both chains were eligible and the winning gain was at least the gain on $T$. If $T$ failed that test, its gain was below $\theta c(e)$, whereas the winning gain was at least $\theta c(e)$. Diminishing returns gives the comparison at $P$. All such charges use distinct accepted gains of $T'$ and hence sum to at most $f(T')\le M$.

	If $T$ never crosses $q$, then $P=T$. Each unassigned optimum item fit in the last scan and has final density at most $\theta_{\rm last}$. Using~\eqref{eq:dla-last},
	\[ f(O\cup P)\le f(P)+f(T')+c(O)\theta_{\rm last}\le 2M+\Gamma/6\le (2+\kappa)M. \]
	Otherwise $c(P)+c(u)>q$. If $u$ is appended in the first scan, all gains through $u$ have density at least $\theta_0=19\Gamma/(6dB)$. Therefore
	\[ M\ge f(P\cup\{u\})\ge\theta_0(c(P)+c(u))>\theta_0 q\ge\opt/[6(1+d)]. \]
	If the crossing occurs in a later scan with density threshold $\theta$, every item of $O$ still unassigned at the snapshot was available throughout the preceding full scan and feasible for $T$. Its density at $P$ is at most the preceding threshold $\theta/(1-d)$. The augmentation of $P$ tested by the algorithm dominates $f(P\cup\{r\})$, since $c(P)+c(r)\le B$. Sum the remaining marginals at $P\cup\{r\}$ using diminishing returns and the earlier cross-chain charge. As $c(K)\le L<(1+d)(c(P)+c(u))$,
	\begin{align*}
	 f(O\cup P)&\le M+f(T')+c(K)\theta/(1-d)\\
	 &\le 2M+\kappa\theta(c(P)+c(u))\le(2+\kappa)M.
	\end{align*}
	The last inequality follows by summing gains through the append of $u$, all at density at least $\theta$. This proves the stated alternatives for both chains since $\kappa M\le\max\{\kappa M,(1-d)\opt/6\}$.
	\end{proof}
	\begin{proof}[Proof of Lemma~\ref{lem:dla2}]
	Here $c(K)\le L\le dB$. Consider the first completed scan with $\theta\le\opt/(6B)$. The initial threshold exceeds $\opt/(6B)$ and~\eqref{eq:dla-last} puts the last one below it. Thus $(1-d)\opt/(6B)\le\theta\le\opt/(6B)$. Let $X_*,Y_*$ be the terminal states of this scan. If either costs at least $(1-d)B$, summing its gains gives $M\ge(1-d)^2\opt/6$, yielding \textbf{e)} and \textbf{g)}. Otherwise each item of $K$ fit in either chain throughout this scan and the earlier ones. With $T$ either terminal state and $T'$ the opposite state, previously assigned members of $K\setminus T$ are charged to their actual gains on $T'$ exactly as above. The remaining unassigned members of $K$ have density at most $\theta$ at $T$. The item $r$ has marginal at most $f(r)\le M$. Consequently
	\[ f(O\cup T)\le f(T)+f(r)+f(T')+\theta c(K)\le 3M+d\opt/6. \]
	The two states are disjoint subsets of the final chains. This upper bound implies \textbf{f)} and \textbf{h)}, including the boundary $L=dB$.
	\end{proof}
\section{Proof of $\RLA$ Algorithm}\label{app:ran-proof}
Fix an optimum $O$, a most expensive $r\in O$, $K=O\setminus\{r\}$, and $L=B-c(r)$. Here $d=\epsilon/10$, $a=\epsilon/100$ and $M$ is the random value of the best feasible candidate. We first analyze the prefix stopped before the coin of the first crossing offer. We then give a separate estimate for $L<aB$. The distinction between the two regimes depends only on the fixed optimum and the input.
\begin{proof}[Proof of Lemma~\ref{lem:rla1}]
Put $P=C_q$. Before each coin is tossed the offer, its gain, and whether it is strictly before the stop are known. If $e\in H$, its contribution to $Z_H$ occurs on tails and its contribution to $f(P)$ on heads. If $e\notin H$, both contributions occur on heads. In each case the conditional expected contributions are equal. A crossing coin contributes to neither sum, including on tails. Summing over at most $n$ offers (or appending zero terms) proves the identity. Membership in $P$ requires an offer strictly before the stop followed by heads, so each item's inclusion probability is at most $1/2$. Apply Lemma~\ref{lem:sub2} to $T\mapsto f(O\cup T)$ with $p=1/2$. The same calculation applies when only the first $m$ offers are counted, for any deterministic $m$. Specifically, if $Z_H^{(m)}$ counts rejected offers in $H$ and accepted offers outside $H$ among those offers, then $\E Z_H^{(m)}=\E f(S_{\min\{m,j\}})$, where each contribution has weight $w_e$. Choosing $H=O$ or $H=O\setminus\{r\}$ gives the two fixed-prefix identities.
\end{proof}
\begin{lemma}[Rounded capacity in the randomized scan]\label{lem:rmoderate}
If $L\ge aB$, Algorithm~\ref{alg:rla} satisfies
\begin{equation}\label{eq:rmoderate}
\opt\le 2\left(1+\frac{1+d}{1-d}+\frac{d}{1-d}+d\right)\E M.
\end{equation}
\end{lemma}
\begin{proof}
Let $q$ be the largest grid point at most $L$ and let $P=C_q$. Write $R_K$ for the total gains of rejected offers from $K$ before the stop, and $A_{\rm in}$ and $A_{\rm out}$ for accepted gains from $K$ and its complement before the stop. Then $Z_K=R_K+A_{\rm out}$ and $A_{\rm in}+A_{\rm out}=f(P)$. Define $\kappa=(1+d)/(1-d)$ and $\nu=d/(1-d)$. The candidate $\operatorname{Aug}(P)$ has value at least $f(P\cup\{r\})$ on every run since $c(P)\le q\le L$.

If an offer $u$ crosses $q$, it cannot do so in the first scan. Every feasible offer in that scan satisfies
\[
 \opt\ge f(P\cup\{u\})\ge \frac{2U}{aB}(c(P)+c(u)),
\]
where $P$ denotes the accepted prefix before this offer. Thus its tentative cost is at most $aB/2<q$. This is where the deterministic bound $U\ge\opt$ is used.

Suppose there is a crossing at a later threshold $\theta$. An optimum item still unoffered just before that coin was available in the preceding scan and fit in the then-current chain because at every such time the chain cost was at most $c(P)\le L$. Its density at $P$ is at most $\theta/(1-d)$. Let $K_0$ contain the members of $K$ not offered strictly before the stop. Previously rejected members of $K$ contribute at most $R_K$ by diminishing returns. Since accepted members in $K$ are $P\cap K$, rounding the cap gives
\begin{align}
c(K_0)&\le L-c(P\cap K)<(1+d)(c(P)+c(u))-c(P\cap K)\notag\\
&=(1+d)(c(P\setminus K)+c(u))+d\,c(P\cap K).\label{eq:roundloss}
\end{align}
Every earlier accepted item had gain at least $\theta$ times its cost. The same holds for $w_u=f(u\mid P)$, regardless of its subsequent coin. Starting with $P\cup\{r\}$ and summing the remaining marginals now gives
\begin{align*}
f(O\cup P)&\le M+R_K+\theta c(K_0)/(1-d)\\
 &\le M+\kappa(Z_K+w_u)+\nu f(P).
\end{align*}
The offer candidate $P\cup\{u\}$ is feasible and recorded, so $w_u\le M-f(P)$.

If no offer crosses $q$, all members of $O$ fit throughout the last scan. The next threshold is below $(1-d)d\Gamma/B$, so the last scanned threshold is less than $d\Gamma/B$. Every unoffered member of $K$ therefore has final density at most $d\Gamma/B$. Thus $f(O\cup P)\le M+R_K+d\Gamma$. As all offer gains are nonnegative and $M\ge f(P),\Gamma$, both branches satisfy the same pathwise inequality
\begin{equation}\label{eq:uniform}
f(O\cup P)\le M+\kappa[Z_K+M-f(P)]+\nu f(P)+d\Gamma.
\end{equation}
Taking expectations only now, Lemma~\ref{lem:rla1} cancels $\E Z_K$ against $\E f(P)$ and lower-bounds the left-hand side by $\opt/2$. This yields \eqref{eq:rmoderate}.
\end{proof}
	\begin{algorithm}[htbp]
	\SetNlSty{text}{}{:}
	\KwIn{$(f,V,B)$, the upper bound $U=19\Gamma$, and $a=\epsilon/100$.}
	$A_1\leftarrow\emptyset$, $A_2\leftarrow\emptyset$, $A_3\leftarrow\emptyset$, $t\leftarrow3U/B$\;
	\ForEach{$e\in V$ with $c(e)\le aB$, in the fixed order}{
		If an index $j$ satisfies $f(e\mid A_j)\ge t c(e)$, append $e$ to the first such $A_j$\;
	}
	$W\leftarrow A_1\cup A_2\cup A_3$; $D\leftarrow\emptyset$, $F\leftarrow W$\;
	\ForEach{$e\in W$, in the fixed order}{
		$x\leftarrow\max\{0,f(D\cup\{e\})-f(D)\}$, $y\leftarrow\max\{0,f(F\setminus\{e\})-f(F)\}$\;
		Set $s\leftarrow x/(x+y)$ if $x+y>0$ and $s\leftarrow1/2$ otherwise\;
		With fresh probability $s$ set $D\leftarrow D\cup\{e\}$; otherwise set $F\leftarrow F\setminus\{e\}$\;
	}
	\Return $D$\;
	\caption{Inexpensive-item subroutine of $\RLA$}\label{alg:cheap}
	\end{algorithm}
\begin{lemma}[Inexpensive items in a feasible union]\label{lem:pool}
Algorithm~\ref{alg:cheap} produces a feasible $W$. For every fixed $K'\subseteq\{e:c(e)\le aB\}$, its returned set $D$ satisfies
\begin{equation}\label{eq:cheap}
 f(K')\le 3\E f(D)+\tfrac32 t\,c(K'),\qquad t=3U/B.
\end{equation}
\end{lemma}
\begin{proof}
We first prove $c(A_j)\le B/3$ after every update. It is true initially. If it holds before a qualifying addition, its tentative cost is at most $B/3+aB<B$, so its new value is at most $\opt$. Telescoping the qualifying gains also gives $f(A_j)\ge t c(A_j)$ for the updated set. Hence $c(A_j)\le\opt/t\le B/3$. By induction and disjointness, $c(W)\le B$.

Let $F^*=\max_{D'\subseteq W}f(D')$. Any $e\in K'\setminus W$ failed the gain test for all three chains at its scan. Its marginal at the final $A_j$ is at most $tc(e)$. For each $j$,
\[ f(K'\cup A_j)\le f((K'\cap W)\cup A_j)+t\,c(K'\setminus W)\le F^*+t\,c(K'). \]
For a uniform $J\in\{1,2,3\}$, disjointness ensures that each item belongs to $A_J$ with probability at most $1/3$. Applying Lemma~\ref{lem:sub2} to $T\mapsto f(K'\cup T)$ and summing the three bounds gives $2f(K')\le3(F^*+t\,c(K'))$.

For completeness, fix a maximizer $D^*\subseteq W$. In double greedy maintain $D\subseteq F$ and track $H=(D^*\cup D)\cap F$. At item $e$, let $a_0=f(D\cup\{e\})-f(D)$ and $b_0=f(F\setminus\{e\})-f(F)$. Submodularity gives $a_0+b_0\ge0$. Set $x=[a_0]_+$, $y=[b_0]_+$. If $e\in H$, only deleting $e$ can cause a loss in the tracked value, at most $x$. If $e\notin H$, only adding $e$ can cause a loss, at most $y$. With the probabilities of the subroutine, the conditional expected loss is at most $xy/(x+y)$. The conditional expected gain in $f(D)+f(F)$ is $(x^2+y^2)/(x+y)$. The first quantity is at most half the second. For $x=y=0$, $a_0=b_0=0$ and the same inequality holds. At termination $D=F=H$. Telescoping yields
\[ F^*-\E f(D)\le\tfrac12[2\E f(D)-f(\varnothing)-f(W)],\qquad F^*\le2\E f(D). \]
Combining the two displays proves \eqref{eq:cheap}. This is the randomized double-greedy argument of Buchbinder et al.\ \cite{usm} on the feasible union constructed above.
\end{proof}
\begin{lemma}[A nearly full-budget optimum item in $\RLA$]\label{lem:rheavy}
If $L<aB$, Algorithm~\ref{alg:rla} satisfies $\opt\le(4+85.5a)\E M$.
\end{lemma}
\begin{proof}
The algorithm calls the subroutine because $c(r)>(1-a)B$. Every item of $K$ costs at most $c(K)\le L<aB$. Apply Lemma~\ref{lem:pool} to $K'=K$. The best feasible singleton is covered by the solution of $\LA$ and hence $f(\{r\})\le \Gamma\le M$. The subroutine returns a candidate $D$, so $f(D)\le M$ and $t=3U/B=57\Gamma/B$. Subadditivity therefore gives
\[
 \opt\le f(\{r\})+f(K)\le 4\E M+\tfrac32(57\Gamma/B)(aB)\le (4+85.5a)\E M.
\]
This includes $c(r)=B$, where $K$ is empty.
\end{proof}
\begin{proof}[Proof of Theorem~\ref{theo:rla}]
For $L\ge aB$, the factor in Lemma~\ref{lem:rmoderate} equals $4+6d/(1-d)+2d$. Since $d=\epsilon/10<1/10$, the excess is at most $(6/0.9+2)\epsilon/10<\epsilon$. For $L<aB$, the excess in Lemma~\ref{lem:rheavy} is $85.5\epsilon/100<\epsilon$. The two cases concern a fixed optimum. Every output candidate is feasible by the offer test, the definition of augmentation, Theorems~\ref{theo:la1}--\ref{theo:lar}, or Lemma~\ref{lem:pool}.

At most $2+\log(38/[ad(1-d)])/[-\log(1-d)]$ threshold scans are made, each with at most $n$ value queries. The value of a tentative offer is already known from its density test. There are at most $1+\lceil\log(1/a)/\log(1+d)\rceil$ caps. At most two $n$-query augmentations are evaluated per cap. The two initial routines take $O(n)$ queries. The inexpensive-item routine is called at most once, with at most three marginal queries per inexpensive item and two per element of $W$, plus initialization. Caching or storing a prefix requires no additional value queries. Substituting $a=\epsilon/100$ and $d=\epsilon/10$ proves the worst-case bound.
\end{proof}
	\section{Additional Experiment Results}
	In this section, we provide 
	an extra experiment to test several scripts. We set $p\in (0.9, 1)$ for Sample Greedy algorithm to archive the best performance in practice as discussed in \cite{Amanatidis_sampleGreedy}. The result is shown in Figure \ref{fig:add-ex}. As can be seen, the solution quality of Sample Greedy has increased a bit, yet its query complexity is almost unchanged. With a small dataset like CIFAR-10, its performance is better than the remaining with B's range from $2\%$ to $\%$ and lower than FANTOM, $\RLA$, and $\DLA$ at other Bs. In the remaining dataset, Sample Greedy is the lowest.  
	Our algorithms, especially $\RLA$,  still take a good query complexity and solution quality consistent with our theoretical analysis. 
	\begin{figure}[hpt]
		\centering
		{\includegraphics[width=.33\linewidth]{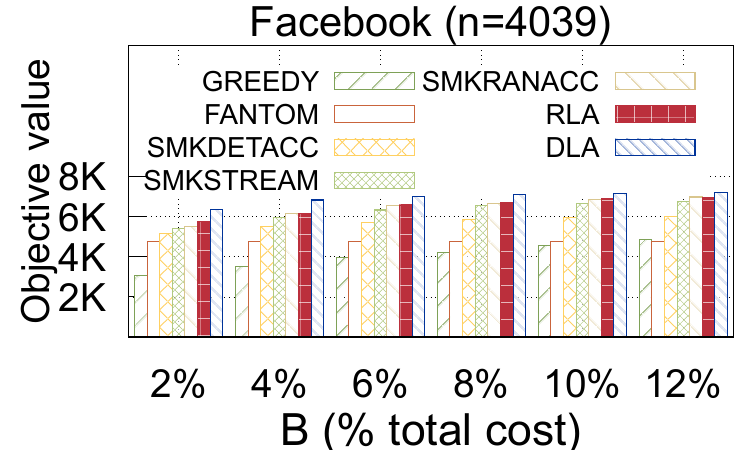}}
		\includegraphics[width=.33\linewidth]{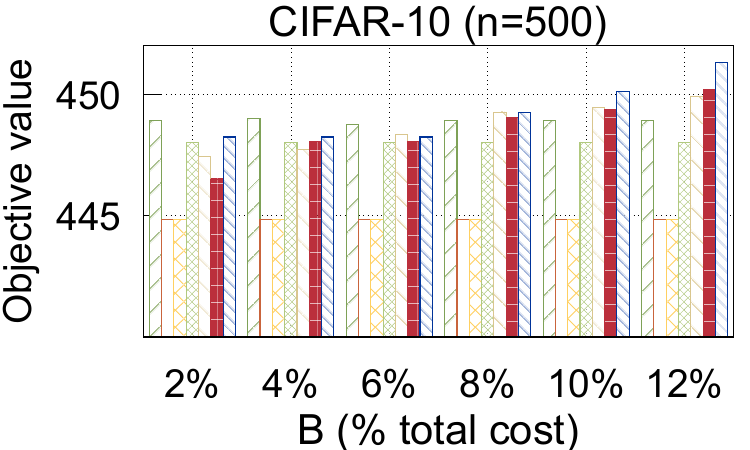}
		\includegraphics[width=.33\linewidth]{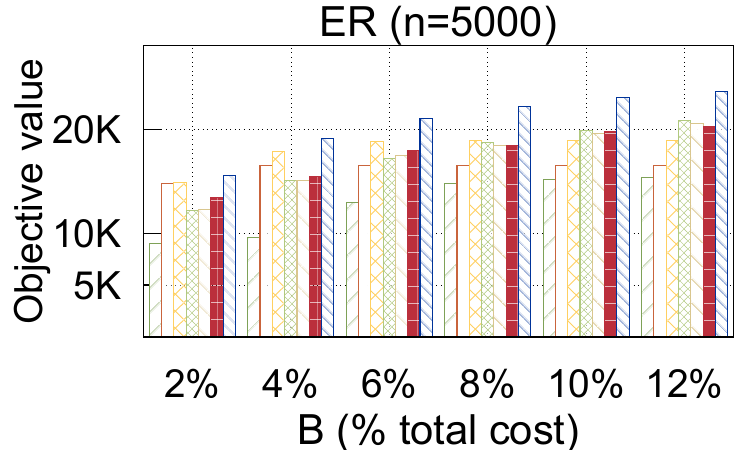}
		\\ 
		\footnotesize{ 	(a) \hspace{5.5cm} (b) \hspace{5.5cm} (c)} 
		\\
		{\includegraphics[width=.33\linewidth]{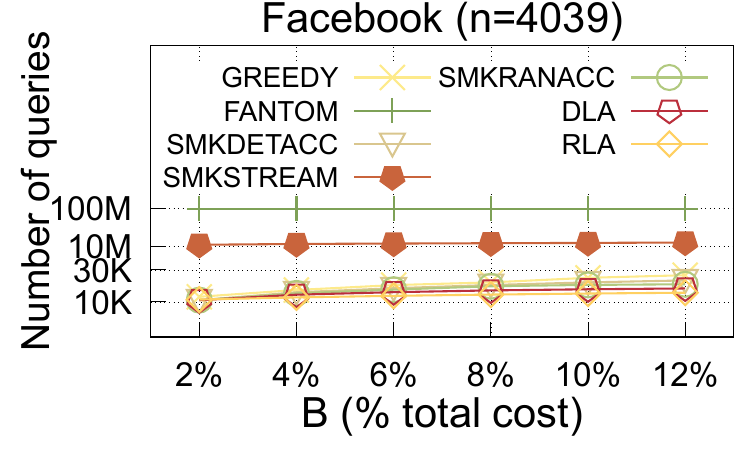}}
		\includegraphics[width=.33\linewidth]{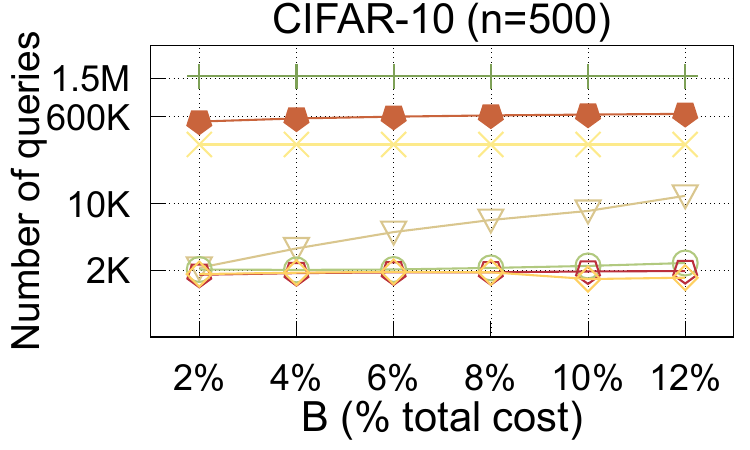}
		\includegraphics[width=.33\linewidth]{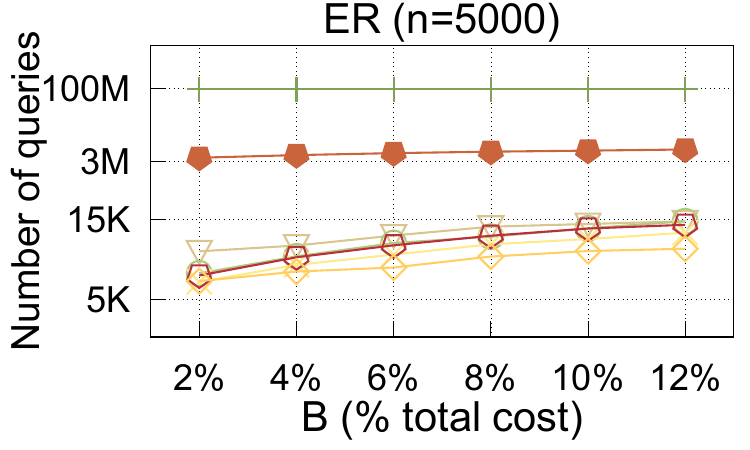}
		\\ 
		\footnotesize{ 	(d) \hspace{5.5cm} (e) \hspace{5.5cm} (f)} 
		\\
		\caption{Performance of algorithms for $\NSMK$ on three instances: (a), (d) Revenue Maximization; (b), (e) Image Summarization and (e), (f)  Maximum Weighted Cut.}
		\label{fig:add-ex}
	\end{figure}
	\section{Additional Related Works}
	\paragraph{Monotone $\SMK$ problem}
	\cite{sub_Wolsey_knap} first studied the  $\SMK$ problem and showed that it was NP-hard but could be approximated within the factor of $(e/(e-1))$.  \cite{sub_knap_orlet,leskovec_celf} first proposed a greedy algorithm, based on making three guesses, achieved an optimal approximation factor of $e/(e-1)$ for monotone $\SMK$ problem; however, they took an expensive query complexity of $O(n^5)$. 
	The recent work \cite{mono-smk-o3} shows an approximation algorithm with the same factor within only $O(n^3 \log n)$ and using $2$ guesses. Since then, many attempts have been made to reduce the query complexity to achieve the optimal approximation factor. Authors \cite{fast_sub} showed a faster approximation algorithm that had $O(n^2(\log(n)/\epsilon)^{1/\epsilon^8})$ but some errors this algorithm was founded in \cite{Ene_HuyLNg_knap}. \cite{Ene_HuyLNg_knap} developed an $(e/(e-1)+\epsilon)$-approximation.  However, their work had to handle complicated multi-linear extensions and required $O((1/\epsilon)^{O(1/\epsilon^{4})}n\log^2 n)$ function evaluations. Another work showed that one could achieve the factor of $2+\epsilon$ for monotone $\SMK$ within $O(nk)$ \cite{Grigory}. However, there is a loophole in the theoretical analysis of this work \cite{Han2021_knap}. 
	
	A streaming fashion model is an excellent approach to solving submodular optimization for massive data since it needs less memory and running time and reduces queries than offline algorithms. Two recent algorithms of \cite{sub_knap_stream} and  \cite{streaming-skm-25} are streaming ones with the same factor of $2.5+\epsilon$ and query complexity  $O(n\log (B)/\epsilon^4)$ in one pass and two passes, respectively where $B$ is the limited cost of the solution. The algorithm of \cite{streaming-smk-multi} reaches a better factor of $2+\epsilon$ but also requires a more significant number of passes and queries of $O(1/\epsilon)$ and $O(n\log^2(B)/\epsilon^8)$, respectively. Especially, \cite{li-linear-knap-nip22} made a breakthrough theoretical by developing a $(2+\epsilon)$-approximation algorithm in a clean linear number of queries $O(n\log(1/\epsilon)/\epsilon)$. They also showed no existing constant factor approximation in $O(n/\log n)$ queries.
	\paragraph{Submodular maximization subject to a cardinality constraint.}
	For the monotone objective function, \cite{Nem_sub} first provided the best approximation algorithm with a factor of $e/(e-1)$ for submodular maximization	subject to a cardinality ($\SMC$) constraint based on a greedy approach and sequential searching. However, the greedy algorithm requires an expensive query complexity of $O(nk)$, where $n$ is the size of the ground set, and $k$ is the size of the cardinality constraint. After that, several algorithmic models have been proposed for improving running time.
	The authors in \cite{fast_sub} are the ones that reduce the running time of a constant approximation algorithm to $O(n \log n)$ for monotone $ \SMC $ problem in a streaming fashion model. In the seminal work, the threshold greedy algorithm was proposed, and it became one of the popular techniques for algorithmic complexity reduction. This work motivated many works to further reduce the number of queries to nearly linear or linear for $\SMC$. 
	\cite{Mirzasoleiman_linearSMC} proposed a stochastic greedy with an approximation algorithm factor of $e/(e-1)$ in $O(n\log(1/\epsilon))$ queries. Both  \cite{Buchbinder14_Submax_Cardinal} and \cite{Fahrbach-mono-sub-soda19} obtain $e/(e-1)$ approximation factor within $O(n)$ queries. However, they are randomized algorithms and may provide poor results in practice.
	Recently, \cite{Kuhnle2021_mono} and \cite{li-linear-knap-nip22} independently proposed a linear-time approximation algorithm with a tight factor of nearly $e/(e-1)$. However, these algorithms cannot work for  non-monotone functions. 
	
	When the objective function is non-monotone, the $\SMC$ is more complex. \cite{lee-car-10} proposed a deterministic local search algorithm with a factor of $4$ in $O(n^4 \log n)$ queries. The authors in \cite{Gupta_nonmono_constrained_submax}  developed an iterated greedy algorithm with the factor of $6$ but significantly reduced the query complexity to $O(nk)$. Recently, \cite{Kuhnle19-nip} developed a $(4+\epsilon)$-approximation algorithm while reduced the queries down to $O(n\log k)$. 
	A randomized algorithm is a popular approach to the non-monotone $\SMC$ problem. In this direction, \cite{Vondrak13} first devised a randomized local search algorithm with an approximation factor of $1/0.309$ in expectation. Later, the randomized greedy approach with an expected factor of $e$ in $O(nk)$ queries was proposed by \cite{Buchbinder14_Submax_Cardinal}. This algorithm was de-randomized with a costly query complexity of $O(k^3n)$ by \cite{Buchbinder-DeranGree} and sped up with factor of $e+\epsilon$ within $O(n\log(1/\epsilon)/\epsilon^2)$ queries \cite{Buchbinder_Qrtradeoff}. The best factor of any algorithm for non-monotonicity $\SMC$ is $1/0.385$ of \cite{BuchbinderF19-bestfactor} using the multi-linear extension technique.
	
	Several algorithms for $\SMC$ focused on two approaches: parallelization and streaming. For parallelization in submodular optimization, the efficiency of a parallel algorithm time is measured by adaptive complexity (or adaptivity) which is the number of sequential rounds an algorithm needs when polynomially multiple queries can be called in parallel in every round. \cite{BalkanskiS18-adapt-stoc18} first applied the adaptive sampling method for the monotone $\SMC$ problem with constant $3$-approximation ratio and $O(\log n)$ adaptivity. After that, \cite{Breuer-mono-sub-icml20} also provided a fast adaptive algorithm whose approximation ratio was arbitrarily close to $e/(e-1)$ for monotone $\SMC$ within  $O(n\log(\log k)$ queries. Especially, \cite{Kuhnle21_nonmono_adapt} contributed two adaptive algorithms for non-monotone $\SMC$  within $O(n\log k)$ query complexity with approximation factors of $5.18+\epsilon$ and $6+\epsilon$ the optimal, respectively. 
	There are many streaming algorithms for $\SMC$ on both monotone \cite{Badanidiyuru_sub_stream2,Norouzi_sub_stream}, etc., and non-monotone \cite{PaulLiu21_nonmono_cardinal}, etc. Authors \cite{Norouzi_sub_stream} improved the approximation ratio by introducing a 2-pass streaming one, SALSA, that used $O(k\log k)$ memory and processed each element using $O(\log k)$ queries. \cite{PaulLiu21_nonmono_cardinal} gave $e/(e-1)+\epsilon$ and $e+\epsilon$ approximation for monotone and non-monotone $\SMC$ respectively, both using $O(k/\epsilon)$ memory but spending $O(n \sqrt{k\log k}/\epsilon)$ queries.
	\section{Notes on the proof corrections}\label{app:corrections}
	\paragraph{Eligibility in $\LA$.} Take $B=10$ and three cost-$5$ items covering respectively $\{a\}$, $\{b\}$ and $\{b,c\}$, with feature weights $100,2,3$. Break the first tie toward $X$. The second item joins $Y$. The third has gain $5$ on $X$ and $3$ on $Y$, but the two acceptance cutoffs are $50$ and $1$. It joins $Y$. The proof of Lemma~\ref{lem:la1} therefore charges items that failed $X$'s test to their costs, keeping Algorithm~\ref{alg:la} intact.
	\paragraph{Dependence in $\LAR$.} On a single cost-$B/2$, unit-value item, let $O_p=O\cap V_p$ and $A$ be the accepted chain. Then $\E[c(O_p)f(A)/B]=p/2>p^2=p\E[f(A)]$ at the stated $p$. Equation~\eqref{eq:lar-fixed-comparator} avoids the product of dependent random variables while retaining the preliminary sample $V_p$.
	\paragraph{Thresholds and stopping.} The prefix grid for $\DLA$ uses the exact base $1+\epsilon'$. The same disjoint-chain scan, density schedule and one-item augmentation are used. For $\RLA$, take $17$ unit-value items, each of cost $B/17$, and a modular objective. With positive probability $V_p=\emptyset$, so $\LAR$ returns a singleton of value one while $\opt=17$. Its expected factor $16.034$ cannot supply an upper bound on this realization. The deterministic interval from $\LA$ supplies that bound. Moreover, deciding whether an offer belongs to a stopped sum after seeing its coin can bias the sum. Lemma~\ref{lem:rla1} tests the tentative cost before the coin, and~\eqref{eq:uniform} combines crossing and noncrossing paths before taking expectations.
\end{document}